\documentclass{article}%{llncs}
\usepackage{amsmath,amssymb,amsfonts,amsthm}
\usepackage{color}
\usepackage{enumerate}
\usepackage{algorithmic,algorithm}
\usepackage{graphics, graphicx}

\newtheorem{thm}{Theorem}[section]%[chapter]
\newtheorem{conjecture}{Conjecture}
\newtheorem{lemma}[thm]{Lemma}
\newtheorem{example}{Example}
\newtheorem*{problem1}{Problem}

\newcommand{\lca}{\ensuremath{\operatorname{lca}}}

\newcommand{\Mo}{\ensuremath{M^\odot}}
\renewcommand{\varnothing}{\odot}
\newcommand{\de}{\ensuremath{\delta}}
\title{On Symbolic Ultrametrics, Cotree Representations, and Cograph Edge Decompositions and Partitions}
\author{Marc Hellmuth \\
Center for Bioinformatics, Saarland University, \\
Building E 2.1, D-66041 Saarbr{\"u}cken, Germany \\
\texttt{mhellmuth@bioinf.uni-sb.de}\\
\and 
Nicolas Wieseke\\
Parallel Computing and Complex Systems Group, \\Department of
Computer Science, Leipzig University, \\
Augustusplatz 10, D-04109 Leipzig, Germany\\
\texttt{wieseke@informatik.uni-leipzig.de}
}
\begin{document}

\maketitle              % typeset the title of the contribution

\begin{abstract}
Symbolic ultrametrics define edge-colored complete graphs $K_n$ and
yield a simple tree representation of $K_n$. We discuss, under which conditions
this idea can be generalized to find a symbolic ultrametric that, in addition,
distinguishes between edges and non-edges of arbitrary graphs $G=(V,E)$ and thus, 
yielding a simple tree representation of $G$. We prove that 
such a symbolic ultrametric can only be defined for $G$ if and
only if $G$ is a so-called cograph. A cograph is uniquely determined by a
so-called cotree. As not all graphs are cographs, we ask, furthermore, 
what is the minimum number of cotrees needed to represent the topology of $G$. 
The latter problem is equivalent to find
an optimal cograph edge $k$-decomposition $\{E_1,\dots,E_k\}$ of $E$ so that each
subgraph $(V,E_i)$ of $G$ is a cograph. An
upper bound for the integer $k$ is derived and it is 
shown that determining whether a graph has a 
cograph 2-decomposition, resp., 2-partition is NP-complete.
%The abstract should summarize the contents of the paper
%using at least 70 and at most 150 words. It will be set in 9-point
%font size and be inset 1.0 cm from the right and left margins.
%There will be two blank lines before and after the Abstract.
%\keywords{Symbolic Ultrametric, Cograph, Cotree, Edge Partition, NP-complete}
\end{abstract}

\sloppy
\section{Introduction}

Given an arbitrary edge-colored complete graph $K_n =(V,E)$ on $n$ vertices, 
B{\"o}cker and Dress \cite{BD98} asked, 
%symmetric map $\delta:X \times X \to M$ that assigns to each pair $(x,y)\in X
%\times X$ a symbol $m\in M$, B{\"o}cker and Dress \cite{BD98} asked, 
whether there is a tree representation of this $K_n$, i.e., 
a rooted tree $T=(W,F)$ with leaf set $V$
%its leaves the vertices contained in $V$
together with a labeling $t$ of the non-leaf vertices in $W\setminus V$ so that 
the least common ancestor $\lca(x,y)$ of distinct leaves 
$x$ and $y$ is labeled with the respective color
of the edge $[x,y]\in E$.
%the first vertex $\lca_{T}(x,y)$ that lies on both unique paths from 
%distinct leaves $x$, resp., $y$ to the root,  % for all distinct $x,y\in V$. 
%$\delta(x,y)$ for all distinct
%$x,y\in X$.
	%$\delta(x,y) =
%t(\lca_{T_{\delta}}(x,y))$. In other words, they asked 
%for a tree $T_{\delta}$ for $\delta$, so
%that the first inner vertex  $\lca_{T_{\delta}}(x,y)$ that lies on both unique paths from $x$,
%resp., $y$ to the root, is labeled with $\delta(x,y)$ for all distinct
%$x,y\in X$. 
The pair $(T, t)$ is then called symbolic representation of the edge-colored graph $K_n$.
%$\delta$ \cite{BD98}. 
The authors showed, 
that there is a 
symbolic representation  $(T,t)$ if and only if the map $\delta$ that assigns 
colors or symbols to the edges in $E$
%symbolic representation for a such map $\delta$ 
fulfills the properties of a so-called symbolic ultrametric \cite{BD98}. 
%that plays a crucial role in phylogenetics \cite{HHH+13, pnas}.
Such maps are crucial 
for the characterization of relationships between genes or proteins, 
so-called orthology relations \cite{HHH+13,HLS+14, LE:15}, 
that lie at the heart of many phylogenomic studies.

Inspired by the work of B{\"o}cker and Dress, we address the following
problem: Does there exist, for an arbitrary undirected graph $G=(V,E)$ a symbolic
ultrametric $\delta:V \times V \to M$ and thus, a symbolic representation
$(T,t)$ of $G$ so that one can distinguish between edges and non-edges of $G$?
In other words, we ask for a coloring $\delta$ of the edges
$[x,y]\in E$, as well as the non-edges $[x,y]\not\in E$, 
so that the topology of $G$ can be displayed by a rooted 
vertex-labeled tree $(T,t)$ s.t.\ for all 
distinct vertices $x,y\in V$ the labeling of the lowest common ancestor 
$\lca(x,y)$ is equal to  $\delta(x,y)$. 
The first result of this contribution provides
that such a symbolic ultrametric 
can only be defined for $G$ if and only if $G$ is a cograph. 
This, in particular, establishes another new characterization of
cographs.
%which does not
%come as a great surprise.

Cographs are characterized by the absence of induced
paths $P_4$ on four vertices. Moreover, Lerchs \cite{lerchs71,lerchs72}
showed that each cograph $G=(V,E)$ is associated with a unique rooted tree $T(G)$,  
called \emph{cotree}.
Obviously, not all graphs are cographs and thus, don't have a cotree
representation. Therefore, we ask for 
the minimum number of cotrees that are needed to represent the 
structure of a given graph $G=(V,E)$ in 
an unambiguous way. As it will turn out, this problem is equivalent 
to find a decomposition $\Pi = \{E_1,\dots, E_k\}$ of $E$
(the elements of $\Pi$ need not necessarily be disjoint)
for the least integer $k$, 
so that each subgraph $G_i = (V,E_i)$,\ $1\leq i\leq k$ is a cograph.
Such a decomposition is called cograph edge $k$-decomposition, 
or cograph $k$-decomposition, for short. 
If the elements of $\Pi$ are in addition pairwise disjoint, we call 
$\Pi$ a cograph $k$-partition. 
We will prove that finding the least integer $k\geq 2$ so that
$G$ has a cograph $k$-decomposition or a cograph $k$-partition
is an NP-hard problem. 
%is NP-complete for $k=2$. 
%In addition, the reduction used for proving NP-completeness
%implies that finding a \emph{cograph edge partition} 
%$\{E_1,E_2\}$  so that  $G_1 = (V,E_1)$ and $G_2 = (V,E_2)$
%are cographs is NP-complete.
Moreover, upper bounds for the integer $k$ for any cograph $k$-decomposition are
derived.
These findings complement results known about so-called 
cograph vertex partitions \cite{Achlioptas199721,GN:02,DMO:14, MSC-Zheng:14}.

\section{Basic Definitions}

\noindent
\emph{Graph.}
In what follows, we consider undirected simple graphs $G=(V,E)$ with vertex
set $V(G) = V$ and edge set $E(G) = E\subseteq \binom{V}{2}$. The \emph{complement graph} $G^c=(V,E^c)$ of
$G=(V,E)$, has edge set $E^c= \binom{V}{2}\setminus E$.
The graph $K_{|V|} = (V,E)$ with $E=\binom{V}{2}$ is called \emph{complete graph}. 
A graph $H=(W,F)$ is an \emph{induced subgraph} of $G=(V,E)$, if $W\subseteq V$
and all edges $[x,y]\in E$ with $x,y\in W$ are contained in $F$.
The \emph{degree} $\deg(v)=|\{e\in E\mid v\in e\}|$ of a vertex $v\in V$
is defined as the number of edges that contain $v$. The maximum degree 
of a graph is denoted with $\Delta$.

\vspace{0.4em}
\noindent
\emph{Rooted Tree.}
A connected graph $T$ is a \emph{tree}, if $T$ does not contain cycles. A
vertex of a tree $T$ of degree one is called a \emph{leaf} of $T$ and all other
vertices of $T$ are called \emph{inner} vertices. 
%An edge of $T$ is an
%\emph{inner} edge if both of its end vertices are inner vertices. 
The set of inner vertices of $T$ is denoted by $V^0$. A \emph{rooted tree}
$T=(V,E)$ is a tree that contains a distinguished vertex $\rho_T\in V$
called the \emph{root}. 
The first inner vertex  $\lca_{T}(x,y)$ that lies on both unique paths from 
distinct leaves $x$, resp., $y$ to the root, is called \emph{most recent
common ancestor of $x$ and $y$}.
%A rooted tree $T$ is called \emph{binary} if each
%inner vertex has outdegree two. 
%The ancestor relation $\preceq_T$ on $V$ is
%the partial order defined, for all $x,y\in V$, by $x \preceq_T y$ whenever
%$y$ lies on the (unique) path from $x$ to the root. 
%For a non-empty subset of leaves
%$L'\subseteq L:=V\setminus V^0$, we define $\lca_T(L')$, or the \emph{most recent
%common ancestor of $L'$}, to be the unique vertex in $T$ that is the
%least upper bound of $L'$ under the partial order $\preceq_T$. In
%case $L'=\{x,y \}$, we put $\lca_T(x,y):=\lca_T(\{x,y\})$. 
If there is no danger of ambiguity, we will write $\lca(x,y)$ rather then	
$\lca_T(x,y)$. 

\vspace{0.4em}
\noindent
\emph{Symbolic Ultrametric and Symbolic Representation.}
In what follows, the set $M$ will always denote a non-empty finite
set, the symbol $\varnothing$ will always denote a special
element not contained in $M$, and $\Mo := M\cup\{\varnothing\}$. 
%The symbol $\varnothing$  corresponds to a
%``non-event''. 
Now, suppose $X$ is an arbitrary non-empty set and 
$\delta:X \times X \to \Mo$ a map.
We call $\delta$ a \emph{symbolic ultrametric}
%\footnote{Note 
%that in \cite{BD98} a symbolic ultrametric is
%defined without the requirement (U0), which
%we have introduced for technical reasons.}
if it satisfies the following conditions:
\begin{itemize}
  \item[(U0)] $\delta(x,y)=\varnothing$ if and only if $x=y$; %\vspace{-0.2cm}
  \item[(U1)] $\delta(x,y) = \delta(y,x)$ for all $x,y \in X$, i.e.\ $\delta$ is
				symmetric; %\vspace{-0.2cm}
  \item[(U2)] $|\{\delta(x,y), \delta(x,z),\delta(y,z)\}| \le 2$
    for all $x,y,z \in X$; and %\vspace{-0.2cm}
  \item[(U3)] there exists no subset $\{x,y,u,v\}\in {X\choose 4}$ such that
%    \begin{equation}
  $    \delta(x,y)= \delta(y,u)  = \delta(u,v) \neq 
      \delta(y,v)= \delta(x,v)  = \delta(x,u).$
 %   \end{equation}
\end{itemize}
Now, suppose that $T=(V,E)$ is a rooted tree with leaf set $X$ 
and that $t:V\to \Mo$ is a map 
such that $t(x)=\varnothing$ for all $x\in X$.
To the pair $(T,t)$ we associate the map 
$d_{(T,t)}$ on $X \times X$ by setting, for all $x,y \in X$,
\begin{equation}
d_{(T,t)}: X \times X \to \Mo;  d_{(T,t)}(x,y) = t(\lca_{T}(x,y)).
\end{equation}
Clearly this map is symmetric and satisfies (U0). We call the pair
$(T,t)$ a \emph{symbolic representation} of a map $\delta:X \times X \to \Mo$, 
if $\delta(x,y)=d_{(T,t)}(x,y)$ holds for all $x,y\in X$. 
For a subset $W\subseteq X\times X$ we denote with $\de(W)$
the restriction of $\de$ to the set $W$. 
%\TODO{discriminating symbolic representation}

\vspace{0.4em}
\noindent
\emph{Cographs and Cotrees.} 
Complement-reducible graph, cographs for short, 
are defined as the class of graphs formed from a single vertex under 
the closure of the operations of union and complementation, namely: (i) 
a single-vertex graph is a cograph;
(ii) the disjoint union of cographs is a cograph; 
(iii) the complement of a cograph is a cograph. 
Alternatively, a cograph can be defined as  a 
$P_4$-free graph
(i.e.\ a graph such that no four vertices induce a subgraph that is a path
of length 3), although there are a number of equivalent characterizations
of such graphs (see e.g.\ \cite{Brandstaedt:99} for a survey).
It is well-known in the literature concerning cographs that,
to any cograph $G$, one can associate a canonical \emph{cotree}
$T(G)=(V,E)$. This is a rooted tree, leaf set equal to the
vertex set $V(G)$ of $G$ and inner vertices that represent so-called "join"
and "union" operations together with a labeling map $t:V^0\to \{0,1\}$ such
that for all $[x,y]\in E(G)$ it holds that $t(\lca(x,y)) = 1$, and
$t(v)\neq t(w_i)$ for all $v\in V^0$ and all children $w_1,\ldots, w_k\in V^0$
of $v$, (cf. \cite{Corneil:81}).

%showed that each cograph $G=(V,E)$ is associated with a unique rooted tree $T(G)$., 
%called \emph{cotree}: Each internal node of $T(G)$ is labeled with either
%$``0''$ or $``1''$, and for each edge $(u,v)$ consisting of inner vertices only, 
%the labels of $u$ and $v$ are different.   
%%and for each edge $(u,v)$ consisting of inner vertices only, 
%%the labels of $u$ and $v$ are different. 
%%, except possibly the root, has at least
%%two children. The inner vertices (non-leaf) of $T(G)$ are labeled with either
%%$``0''$ or $``1''$ and for each edge $(u,v)$ consisting of inner vertices only, 
%%the labels of $u$ and $v$ are different. 
%The leaves of $T(G)$ are the vertices contained in $V$ and
%two vertices $u,v\in V$ form an edge $(u,v)\in E$ if and only
%the lowest common ancestor $\lca_{T(G)}(u,v)$ of $u$ and $v$ in $T(G)$ is labeled with $``1''$.
%In other words, the symbolic ultrametric $\delta:V \times V \to M$
%maps edges to the symbol $``1''$ and non-edges to $``0''$. 

\vspace{0.4em}
\noindent
\emph{Cograph $k$-Decomposition and Partition, and Cotree Representation.}
Let  $G=(V,E)$ be an arbitrary graph. 
A decomposition $\Pi=\{E_1,\dots E_k\}$ of $E$ is a called \emph{(cograph) $k$-decomposition}, 
if each subgraph $G_i = (V,E_i)$,\ $1\leq i\leq k$ of $G$ is a cograph.
We call $\Pi$ a \emph{(cograph) $k$-partition} if $E_i\cap E_j=\emptyset$, for all 
distinct $i,j\in \{1,\dots,k\}$.
A $k$-decomposition $\Pi$ is called \emph{optimal}, if $\Pi$ has the least number $k$
of elements among all cograph decompositions of $G$. Clearly, for a cograph only
k-decompositions with $k=1$ are optimal.
A $k$-decomposition $\Pi=\{E_1,\dots E_k\}$ is \emph{coarsest}, if no elements
of $\Pi$ can be unified, so that the resulting decomposition is a 
cograph $l$-decomposition, with $l< k$. In other words, $\Pi$ is 
coarsest, if for all subsets $I\subseteq \{1,\dots,k\}$ with $|I|>1$
it holds that $(V, \cup_{i\in I} E_i)$ is not a cograph. 
Thus, every optimal $k$-decomposition is also always a coarsest one. 

A graph $G = (V,E)$ is \emph{represented by a set of cotrees} $\mathbb T = \{T_1,\dots,T_k\}$, 
each $T_i$ with leaf set $V$, if and only if for each edge $[x,y]\in E$ there is a tree $T_i\in \mathbb T$
with $t(\lca_{T_i}(x,y)) =1$.

\vspace{0.4em}
\noindent
\emph{The Cartesian (Graph) Product}
$G\Box H$ has vertex set $V(G\Box H)=V(G)\times
V(H)$; two vertices $(g_1,h_1)$, $(g_2,h_2)$ are adjacent in $G\Box H$ if
$[g_1,g_2]\in E(G)$ and $h_1=h_2$, or $[h_1,h_2]\in E(H)$ and $g_1 =
g_2$. It is well-known that the Cartesian product is associative, 
commutative and that the single vertex graph $K_1$ serves as unit element
\cite{Hammack:11a}. Thus, the product $\Box_{i=1}^n G_i$ 
of arbitrary many factors $G_1,\dots, G_n$ is well-defined. 
For a given product $\Box_{i=1}^n G_i$, we define 
%assosciativ, commutativ ..
the $G_i$-layer $G_i^w$ of $G$ (through vertex $w$ that has coordinates $(w_1,\dots,w_n)$) 
as the induced subgraph with vertex set
$V(G_i^w)=\{v = (v_1,\dots,v_n)\in \times_{i=1}^n V(G_i) \mid v_j=w_j, \text{ for all } j\neq i\}$. 
%It is
%isomorphic to $G_i$.
Note, $G_i^w$ is isomorphic to $G_i$ for all $1\leq i\leq n$, $w\in V(\Box_{i=1}^n G_i)$. 
The \emph{$n$-cube} $Q_n$ is the Cartesian product
$\Box_{i=1}^n K_2$. %, where $K_2$ denotes the complete graph on $2$ vertices.

\section{Symbolic Ultrametrics}

Symbolic ultrametrics and respective representations as event-labeled
trees, have been first characterized by B{\"o}cker and Dress \cite{BD98}.

\begin{thm}[\cite{BD98,HHH+13}]\label{bd}
Suppose $\delta: V \times V\to\Mo$ is a map. Then there is a 
%discriminating 
symbolic representation of $\delta$ 
if and only if $\delta$ is a symbolic ultrametric. 
Furthermore, this representation 
can be computed in polynomial time. 
\end{thm}

Let $\delta:V \times V \to \Mo$ be a map satisfying 
Properties (U0) and (U1). For each
fixed $m\in M$, we define an undirected graph $G_m:=G_m(\delta)=(V,E_m)$ 
with edge set
\begin{equation}
E_m =  
\left\{ \{x,y\} \mid \delta(x,y)=m, \,\, x,y \in V \right\}.
\end{equation}
Thus, the map $\delta$ can be considered as an edge coloring of 
a complete graph $K_{|V|}$, where each edge $[x,y]$ obtains
color $\delta(x,y)$.
Hence, $G_m$ denotes the subgraph of the edge-colored graph 
$K_{|V|}$, that contains all edges colored with $m\in M$. 
The following result establishes the connection between
symbolic ultrametrics and cographs. 

\begin{thm}[\cite{HHH+13}]\label{thm:cograph}
  Let $\delta:V\times V\to \Mo$ be a map satisfying Properties 
  (U0) and (U1). Then 
  $\delta$ is a symbolic ultrametric if and only if 
  \begin{itemize}
  \item[(U2')] For all $\{x,y,z\}\in\binom{V}{3}$ there is an $m\in M$ 
	such that
    $E_m$ contains two of the three edges $\{x,y\}$, $\{x,z\}$, and
    $\{y,z\}$.
  \item[(U3')] $G_m$ is a cograph for all $m\in M$. 
  \end{itemize}
\end{thm}
%Thus, an edge-coloring defined by the map $\delta$ yields 
%a symbolic ultrametric if and only if the respective edge-colored 
%subgraphs $G_m$, $m\in M$ of $K_{|V|}$ satisfy Property (U2') and (U3').

Assume now, we have given an arbitrary  subgraph
$G=(V,E)\subseteq K_{|V|}$. Let $\delta$ be a map defined on 
$V \times V$ so that edges $e\in E$ obtain a different color
then the non-edges $e\in E(K_{|V|})\setminus E$ of $G$. 
The questions then arises, whether such a map fulfilling
the properties of symbolic 
ultrametric can be defined and 
thus, if there is tree representation $(T,t)$ of $G$. 
Of course, this is possible only if 
$\delta$ restricted to $E$, resp., $E^c$ is a symbolic ultrametric, 
while it also a symbolic ultrametric 
on the complete graph $K_{|V|}=(V,E\cup E^c)$. 
%Therefore, for a given arbitrary graph $G=(V,E)$ one might ask 
%if there is a set of symbols $M$ s.t. $\delta:V\times V \to \Mo$ is an 
%symbolic ultrametric and in addition, 
%edges and non-edges of $G$ can be distinguished by their assigned symbols. 
%In other words,  while it must 
%be 
The next theorem answers the latter question and, in addition, 
provides a new characterization of cographs.

\begin{thm}
	Let $G=(V,E)$ be an arbitrary 
	(possibly disconnected) graph, $W= \{(x,y)\in V\times V \mid [x,y]\in E\}$
	and $W^c= \{(x,y)\in V\times V \mid [x,y]\not\in E\}$.
	There is a symbolic ultrametric $\delta:V\times V\to \Mo$
	s.t.\ $\de(W)\cap \de(W^c)=\emptyset$ 
%	Then  $\de$ is a symbolic ultrametric 
	if and only if 
	$G$ is a cograph.
\end{thm}
\begin{proof}
	First assume that $G$ is a cograph. Set $\de(x,x)  = \varnothing$ for 
	all $x\in V$ and set $\de(x,y)=\de(y,x)=1$ if  
	$[x,y]\in E$ and, otherwise, to $0$. 
	Hence, condition $(U0)$ and $(U1)$ are fulfilled. 
	Moreover, by construction $|M|=2$ and thus, condition $(U2')$ is trivially
	fulfilled. Furthermore, since $G_1(\de)$ and its complement $G_0(\de)$
	are cographs, $(U3')$ is satisfied. Theorem \ref{thm:cograph} implies that
	$\de$ is a symbolic ultrametric. 
	
	Now, let $\delta:V\times V\to \Mo$ be a symbolic ultrametric with 
	$\de(W)\cap \de(W^c)=\emptyset$. 
	Assume for contradiction that $G$ is not a cograph. Then $G$ contains 
	an induced path $P_4 = a-b-c-d$. Therefore, at least one edge $e$
	of this path $P_4$
	must obtain a color $\de(e)$ different from the other two edges
	contained in this $P_4$, 
	as otherwise $G_{\de(e)}(\de)$ is not a cograph and 
	thus, $\de$ is not a symbolic ultrametric (Theorem \ref{thm:cograph}).
	For all such possible maps $\de$ ``subdividing'' this 
	$P_4$ we always obtain that two edges of at least one of 
	the underlying paths $P_3 = a-b-c$ or $b-c-d$ 
	must have different colors. W.l.o.g. assume that
	$\de(a,b)\neq \de(b,c)$. Since $[a,c]\not\in E$
	and $\de(W)\cap \de(W^c)=\emptyset$ we can conclude that
	$\de(a,c)\neq \de(a,b)$ and 	$\de(a,c)\neq \de(b,c)$. 
	But then condition $(U2')$ cannot be satisfied, and 
    Theorem \ref{thm:cograph} implies that $\de$ is not a
	symbolic ultrametric. \hfill \qed
\end{proof}

The latter result implies, that there is no hope for finding a map $\delta$
for a graph $G$, that assigns symbols or colors to edges, resp., non-edges such that
for $\delta$ (and hence, for $G$) there is a symbolic representation $(T,t)$, unless $G$ is
already a cograph. In other words, every symbolic representation 
  $(T,t)$ for an arbitrary graph $G$ (which only exists
if $G$ is a cograph) is a cotree. 
However, this result does not come as a big surprise, as a
cograph $G$ is characterized by the existence of a unique (up to isomorphism)
cotree $(T,t)$ representing the topology of $G$. 
%If $G$ is a cograph, then  $(T,t)$  is the usual cotree. 
The (decision version of
the) problem to edit a given graph $G$ into a cograph $G'$, 
and thus, to find the closest graph $G'$ that has 
a symbolic representation, is NP-complete \cite{Liu:11,Liu:12}.
In this contribution, we are interested in 
%This leads us to 
the following problem: 
%We are interested 
%in the following  problem: 
\emph{What is the minimum number of cotrees
that are needed to represent the topology of $G$ in an unambiguous way?}
%\TODO{connection to symb umetric besser rausarbeiten}

\section{Cotree Representation and Cograph $\boldsymbol{k}$-Decomposition}

%\subsection{Coarsest and Min. Cograph-$\boldsymbol{k}$-Decompositions }

%Before we start to examine the computational complexity of the \textsc{Cograph $k$-Decomposition}
%problem, we discuss 
%\TODO{this at the end for open problems .. }

%Obviously, any minimum $k$-decomposition must
%also be a coarsest $k$-decomposition. 

Recollect, a graph $G = (V,E)$ is represented by a set of cotrees $\mathbb T = \{T_1,\dots,T_k\}$, 
if and only if for each edge $[x,y]\in E$ there is a tree $T_i\in \mathbb T$
with $t(\lca_{T_i}(x,y)) =1$. Note, by definition, each cotree $T_i$ determines a subset 
$E_i = \{[x,y]\in E \mid t(\lca_{T_i}(x,y))=1\}$ of $E$. Hence, the subgraph 
$(V,E_i)$ must be a cograph. Therefore, in order to find the minimum number of cotrees
representing a graph $G$, we can equivalently ask for a decomposition 
$\Pi=\{E_1,\dots,E_k\}$ of $E$ so that each subgraph $(V,E_i)$ is a cograph, 
where $k$ is the least integer among all cograph decompositions of $G$. 
Thus, we are dealing with the following two equivalent problems.

\begin{figure}[tbp]
  \centering
  \includegraphics[bb= 182 515 409 694, scale=0.6]{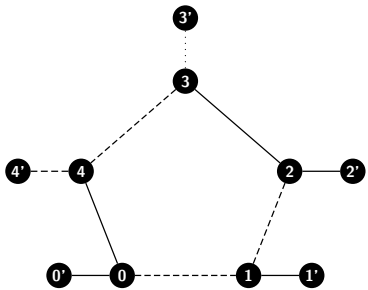} %\hfill
  \caption{Full enumeration of all possibilities (which we leaf to the reader), 
			shows that the depicted graph has no cograph 2-decomposition. However, it has 
			a cograph 3-decomposition that is
			highlighted by dashed-lined, dotted and bold edges.}
  \label{fig:3col}
\end{figure}

\begin{problem1}\textsc{Cotree $k$-Representation}

\vspace{0.3em}
\begin{tabular}{ll}
	\emph{Input:}& Given a graph $G=(V,E)$ and an integer $k$ . \\
	\emph{Question:}& Can $G$ be represented by $k$ cotrees?
\end{tabular}
\end{problem1}

\begin{problem1}\textsc{Cograph $k$-Decomposition}

\vspace{0.3em}	
\begin{tabular}{ll}
  \emph{Input:}& Given a graph $G=(V,E)$ and an integer $k$. \\
\emph{Question:}& Is there a cograph $k$-decomposition of $G$?
\end{tabular}
\end{problem1}

Clearly, any cograph has an optimal $1$-decomposition, while
for cycles of length $>4$ or paths $P_4$ there is always an 
optimal cograph 2-decomposition. 
However, there are examples of graphs that do not have
a 2-decomposition, see Figure \ref{fig:3col}.
To derive an upper bound for the integer $k$ s.t.\ there 
is a cograph $k$-decomposition for arbitrary graphs, 
the next theorem is given. 

\begin{thm}
For every graph $G$ with maximum degree $\Delta$
there is a cograph $k$-decomposition with 
$1\leq k\leq \Delta+1$ that can 
be computed in $O(|V||E| + \Delta(|V|+|E|))$ time.
Hence, any graph can
be represented by at most $\Delta+1$ cotrees.
\end{thm}
\begin{proof}
Consider a proper edge-colorings $\varphi:E\to \{1,\dots,k\}$ of $G$, 
i.e., an edge coloring such that no two incident edges obtain the same color.
Any proper edge-coloring using $k$ colors 
yields a cograph $k$-partition $\Pi= \{E_1,\dots,E_k\}$ where
$E_i=\{e\in E\mid \varphi(e)=i\}$,
because any connected component in $G_i =(V,E_i)$ 
is an edge and thus, no $P_4$'s are contained in $G_i$.
%since each subgraph 
%$G_i=(V,E_i)$ is a cograph as any connected component in $G_i$ 
Vizing's Theorem \cite{V:64} implies that for each graph 
there is a proper edge-coloring using $k$ colors with $\Delta\leq k\leq \Delta+1$. 
%The resulting partition $\Pi$ of the edge set is a cograph $k$-partition, and thus, 
%a cograph $k$-decomposition.   

An proper edge-coloring using at most $\Delta+1$ colors can be computed 
with the Misra-Gries-algorithm in $O(|V||E|)$ time \cite{MG:92}. 
Since the (at most $\Delta+1$) respective cotrees can be constructed in linear-time $O(|V|+|E|)$ \cite{corneil1985linear}, we
derive the runtime $O(|V||E| + \Delta(|V|+|E|))$. \hfill \qed
\end{proof}

%\begin{thm}
%Maybe: add Graph classes with $\Delta$ colors (thus $\Delta$ cotrees at most)in polytime - class1 edge coloring classes.
%\end{thm}

%\TODO{check carefully correct use of partition and decomposition}

Obviously, any optimal $k$-decomposition must also be a coarsest
$k$-decomposition, while the converse is in general not true, see
Fig.\ref{fig:pd}.
The partition $\Pi= \{E_1,\dots,E_k\}$ 
obtained from a proper edge-coloring is usually not a coarsest one, 
as possibly $(V,E_J)$ is a cograph, where
$E_J=\cup_{i\in J} E_i$ and  $J\subseteq \{1,\dots,l\}$.
A graph having an optimal cograph $\Delta$-decomposition
is shown in Fig.\ \ref{fig:3col}. 
Thus, the 
derived bound $\Delta+1$ is almost sharp.
Nevertheless, we assume that this bound can be sharpened: 

\begin{conjecture}
For every graph $G$ with maximum degree $\Delta$
there is a cograph $\Delta$-decomposition.
	\label{conj:k-partDelta}
\end{conjecture}

%We first show that coarsest $k$-decompositions with  
%$k\leq \Delta+1$ can be determined in polynomial time, 
%see  Algorithm \ref{alg:opt}. For the sake of simplicity, 
%we call for a graph $G$ with cograph-k-decomposition
%the \emph{connected components w.r.t. to $M_i\in \Pi$}
%the connected components of $G$ consisting of 
%$i$ colored edges only.
%\end{proof}
%By the latter result, we can immediately infer the following theorem.
%\begin{thm}
%Every graph can be represented by at most $\Delta+1$ cotrees.
%\end{thm}
%\TODO{somewhere statement that minimal and thus, coarsest are bounded, ie, $k\leq \Delta+1$}

\begin{figure}[t]
  \centering
  \includegraphics[bb= 6 416 435 651, scale=0.6]{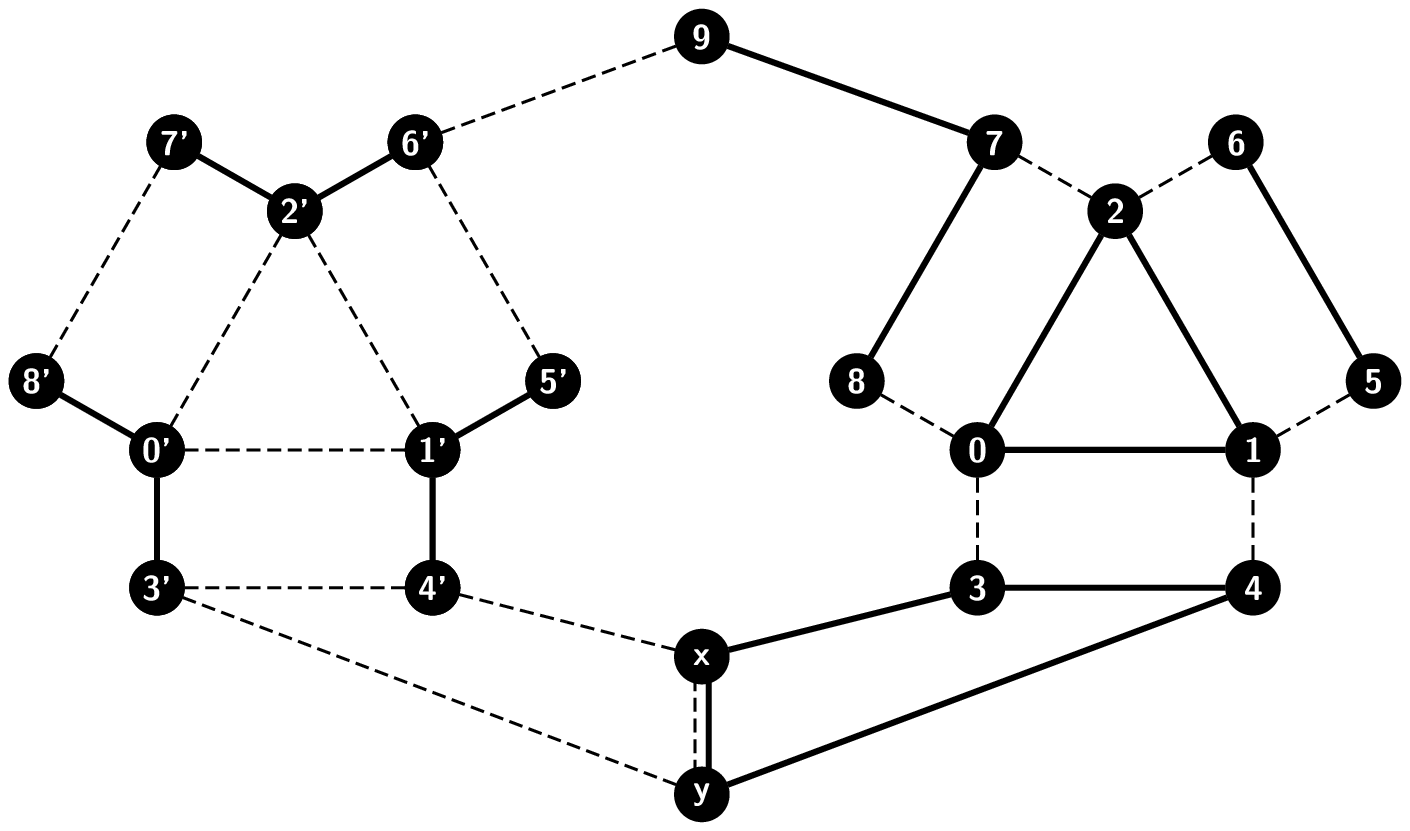}\\[0.3cm]  \centering
  \includegraphics[bb= 72 497 561 609, scale=0.7]{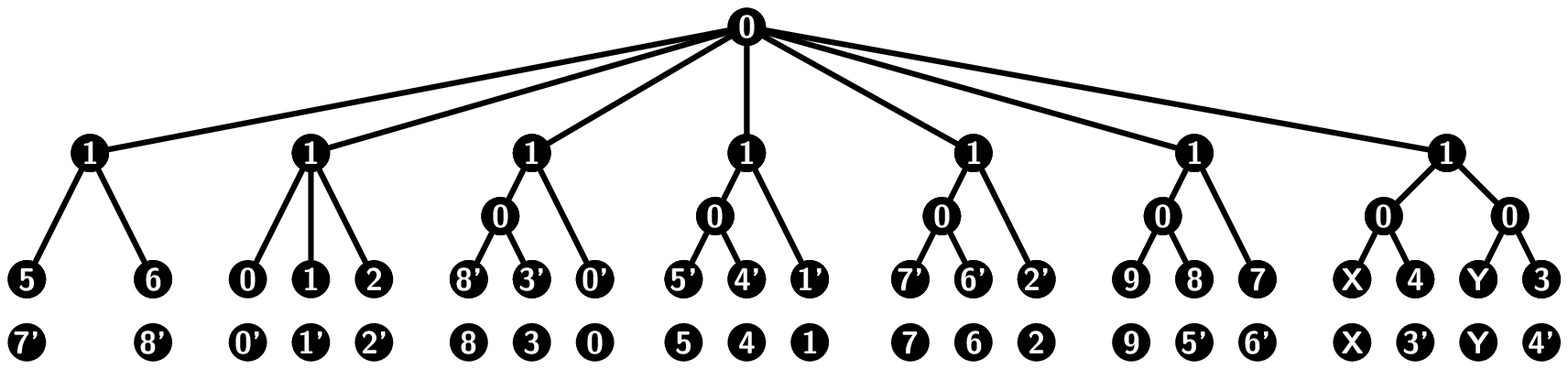}
  \caption{The shown (non-co)graph $G$  has a 2-decomposition $\Pi=\{E_1,E_2\}$.
						Edges in the different elements $E_1$ and $E_2$ are 
            highlighted by dashed and solid edges, respectively. Thus, 
            two cotrees, shown in the lower part of this picture, are sufficient to represent the structure of $G$.
						The two cotrees are isomorphic, and thus, differ only in the arrangement of their leaf sets. For this
		  	   reason, we only depicted one cotree with two different leaf sets. 
					  Note,  $G$ has no 2-partition, but a coarsest 3-partition. The latter can easily be
            verified by application of the construction in Lemma \ref{lem:col-lit}. 
}
  \label{fig:pd}
\end{figure}

However, there are examples of non-cographs containing many induced $P_4$'s that have 
a cograph $k$-decomposition with $k\ll\Delta+1$, which implies that
any optimal $k$-decomposition of those graphs 
will have significantly less elements than $\Delta+1$, see the following
examples. 

\begin{example}
	Consider the graph $G = (V,E)$ with vertex set $V=\{1,\dots,k\}\cup\{a,b\}$
	and $E =\{[i,j]\mid i,j\in \{1,\dots,k\}, i\neq j\}\cup\{[k,a],[a,b]\}$.
	The graph $G$ is not a cograph, since there are induced $P_4$'s of the form
	$i-k-a-b$, $i\in \{1,\dots,k-1\}$. On the other hand, the subgraph
	$H=(V, E\setminus \{[k,a]\})$ has two connected components, 
	one is isomorphic to the complete graph $K_k$ on $k$ vertices and the 	
	other to the complete graph $K_2$. Hence, $H$ is a cograph. 
	Therefore, $G$ has a cograph 2-partition $\{E\setminus \{[k,a]\}, \{[k,a]\}\}$, independent from $k$
	and thus, independent from the maximum degree $\Delta = k$. 
\end{example}

\begin{example}
 Consider the 2n-dimensional hypercube $Q_{2n}=(V,E)$ with maximum degree $2n$. 
 We will show that this hypercube has a
 coarsest cograph $n$-partition $\Pi=\{E_1,\dots,E_n\}$, which implies that for any optimal
 cograph $k$-decomposition of $Q_{2n}$ we have $k\leq \Delta/2$. 

 We construct now a cograph $n$-partition of $Q_{2n}$. Note, $Q_{2n} =
\Box_{i=1}^{2n} K_2 = \Box_{i=1}^n (K_2\Box K_2) = \Box_{i=1}^n Q_2$. In
order to avoid ambiguity, we write $\Box_{i=1}^n Q_2$ as $\Box_{i=1}^n
H_i$, $H_i\simeq Q_2$ 
and assume that $Q_2$ has edges $[0,1]$, $[1,2]$,
$[2,3]$, $[3,0]$. 
The cograph $n$-partition of $Q_{2n}$ is defined as $\Pi
= \{E_1,\dots, E_n\}$, where $E_i=\cup_{v\in V} E(H_i^v)$. In other words,
the edge set of all $H_i$-layers in $Q_{2n}$  
 constitute a single class $E_i$ in
the partition for each $i$. 
Therefore, the subgraph
$G=(V,E_i)$ consists of $n$ connected components, each component is isomorphic
to the square $Q_2$. Hence, $G_i=(V,E_i)$ is a cograph. 

Assume for contradiction that $\Pi = \{E_1,\dots, E_n\}$ is not a coarsest partition.
Then there are distinct classes $E_i$, $i\in I\subseteq \{1,\dots,n\}$ such
that $G_I=(V,\cup_{i\in I}E_i)$ is a cograph. W.l.o.g.\ assume that $1,2\in
I$ and let $v=(0,\dots,0)\in V$. Then, the subgraph $H_1^v\cup H_2^v
\subseteq Q_{2n}$ contains a path $P_4$ with edges $[x,v]\in E(H_1^v)$
and $[v,a], [a,b]\in E(H_2^v)$, where x=(1,0,\dots,0), a=(0,1,0\dots,0)
and $b=(0,2,0\dots,0)$. 
%	W.l.o.g. let $x_1=(1,0,\dots,0)$, $x_2=(0,1,0,\dots,0)$ and $y_2=(0,2,0,\dots,0)$. 
	By definition of the Cartesian product, there are no edges connecting
$x$ with $a$ or $b$ or $v$ with $b$ in $Q_{2n}$ and thus, this path $P_4$ is induced.
As this holds for all subgraphs $H_i^v\cup H_j^v$ ($i,j\in I$ distinct) and
thus, in particular for the graph $G_I$ we can conclude that classes of
$\Pi$ cannot be combined. Hence $\Pi$ is a coarsest cograph $n$-partition. 
% Note, $Q_{2n}$ has ... (exponen?) many induced $P_4$'s.
\label{ex:hq}
\end{example}

Because of the results of computer-aided search for $n-1$-partitions and
decompositions of hypercubes $Q_{2n}$ 
we are led to the following conjecture:

\begin{conjecture}
	Let $k\in \mathbb N$ and $k>1$.
	Then the $2k$-cube has no cograph $k-1$-decomposition, i.e., 
	the proposed $k$-partition of the hypercube $Q_{2k}$  in Example \ref{ex:hq}  is also optimal.
	\label{conj:k-part2}
\end{conjecture}

The proof of the latter hypothesis would immediately verify the
next conjecture. 

\begin{conjecture}
	For every $k\in \mathbb N$ there is a graph that has
	an \emph{optimal} cograph $k$-decomposition. 
%	a  but no cograph $k-1$-partition, i.e., for 
%	this graph $G$ 
%	any $k$-partition is also optimal.
	\label{conj:k-part1}
\end{conjecture}

Proving the last conjecture appears to be difficult.
We wish to point out that there is  a close relationship
to the problem of finding pattern avoiding words, see e.g.\
\cite{Br:05,BM:08,P:08,pudwell2008enumeration,bilotta2013counting,bernini2007enumeration}:
Consider a graph $G=(V,E)$ and an ordered list $(e_1,\dots,e_m)$ of the edges $e_i\in E$. 
We can associate to this list  $(e_1,\dots,e_m)$ a word $w=(w_1,\dots,w_m)$. 
By way of example, assume that we want to find a valid cograph 2-decomposition $\{E_1,E_2\}$ of $G$
and that $G$ contains an induced $P_4$ consisting of the edges $e_i,e_j,e_k$. 
Hence, one has to avoid assignments of the edges $e_i,e_j,e_k$
to the single set $E_1$, resp., $E_2$. The latter is equivalent to
find a binary word $(w_1,\dots,w_m)$  such that 
$(w_i,w_j,w_k) \neq (X,X,X)$, $X\in\{0,1\}$ for each of
those induced $P_4$'s.
%Analogously, for each induced square $Q_2$ of $G$  	containing edges 
%$e_i,e_j,e_k,e_l$ 
%one can equivalently ask for a binary word $w=(w_1,\dots,w_m)$
%such that  $(w_i,w_j,w_k, w_l) \not\in \{(X,X,X,Y), (X,X,Y,X), (X,Y,X,X), (Y,X,X,X)\}$, 
%where $X\neq Y$ and $X,Y\in \{0,1\}$. 
The latter can easily be generalized
to find pattern avoiding words over an alphabet $\{1,\dots,k\}$ to get 
a valid $k$-decomposition. 
However, to the authors knowledge, results concerning the counting
of $k$-ary words avoiding forbidden patterns and thus, verifying if there is any such word
(or equivalently a $k$-decomposition) are basically known for scenarios
like:
If  $(p_1,\dots p_l) \in \{1,\dots,k\}^l$ (often $l<3$), 
then  none of the words $w$ that contain
a subword $(w{_{i_1}},\dots,w{_{i_l}})=(p_1,\dots p_l)$ 
with $i_{j+1}=i_{j}+1$ (consecutive letter positions) or
$i_j<i_k$ whenever $j<k$ (order-isomorphic letter positions)  is allowed. 
%Slight generalizations to so-called order isomorphic words are established in \cite{BM:08}
However, such findings are to restrictive to our problem, since we are looking
for words, that have only on a few, but fixed positions of non-allowed patterns.
%, as 
%we 
%% ∈ {0, 1}ℓ, i.e., a word ω is contained in F[p] if and only if there is no sequence
%%of consecutive indices i, i + 1, . . . , i + ℓ − 1 such that ωiωi+1 · · · ωi+ℓ−1 = p0p1 · · · pℓ−1.
%is known for patterns $(w_1,w_2,\dots,w_n)$ of the form that in none of
%the words it is allowed to have the pattern $(w_1,w_2,\dots,w_n)$, ....
%and not only restricted to a few positions only ,, 
Nevertheless, we assume that results concerning the recognition of 
pattern avoiding words might offer an avenue to solve the latter conjectures.

%\TODO{  Relation to "find sequences with forbidden patterns" .. 
%However, we assume that $Q_n$ is a good candiate as it has among all graphs
%with fixed vertics and edges cardinal. the largest number of induced $P_4$ - 
%maybe show this?} 

%\begin{lem}
%	Given $G$ with minimal cograph $k_G$-partition
%	and $H$ with minimal cograph $k_H$-partition. 
%	For any minimal cograph $k_{G\Box H}$-partition of $G\Box H$ holds
%	that $k_{G\Box H} \leq k_G+k_H$.  
%%	\TODO{Does this also hold for all minimal?}
%\end{lem}
%\begin{proof}
%	Assume 	$G$ an optimal cograph $k_G$-partition $\Pi_G$
%	and $H$ an optimal cograph $k_H$-partition $\Pi_H$. 
%	Then  all $G$-layers in $\Pi_G$ and all $H$-layers in 
%	$\Pi_H$ .. 
%\end{proof}

%\subsection{P4 sparse have minimal 2-partition can be recog, in linear time}

%A graph $G$ is called $P_4$-sparse if every $5$-node subgraph contains at most
%one induced path on four vertices $P_4$ \cite{Hoang:85}. In \cite{Jamison:89,Jamison:92}
%it is shown that 
%a graph is $P_4$-sparse if and only if exactly one of the following three
%alternatives is true for every induced subgraph $H$ of $G$: (i) $H$ is not
%connected, (ii) $\overline{H}$ is not connected, or (iii) $H$ is a spider.

\subsection{NP-completeness and NP-hardness Results}

We are now in the position  to prove the NP-completeness of
\textsc{Cotree 2-Representation} and \textsc{Cotree 2-Decomposition}. 
These results allow to show that the problem of determining whether
there is cograph 2-partition is NP-complete, as well. 

We start with two lemmata concerning cograph 2-decompositions of the
graphs shown in Fig.\ \ref{fig:literal} and \ref{fig:clause}.

%Complexity of Determining Optimal Cograph-$\boldsymbol{k}$-Decompositions}

%\TODO{to prove that for bool-formula graph any decomposition is a partition}

%We show that the following problem is NP-complete.  ..
%\begin{probl}\textsc{2-Cotree Representation}\\
%\begin{tabular}{ll}
%	$\ \ \ $  \emph{Input:}&Given a graph $G=(V,E)$. \\
%	$\ \ \ $\emph{Question:}&Can $G$ be represented by two cotrees?
%\end{tabular}
%\end{probl}

%The latter problem is equivalent to 

%\begin{probl}\textsc{Cograph 2-Decomposition}\\
%\begin{tabular}{ll}
%	$\ \ \ $  \emph{Input:}&Given a graph $G=(V,E)$. \\
%	$\ \ \ $\emph{Question:}&Is there a Cograph 2-Decomposition of $G$?
%\end{tabular}
%\end{probl}

%Note, the NP-completeness of \textsc{Cograph 2-Decomposition}
%immediately implies that the optimization problem to find 
%the smallest $k$ such that there is a cograph-k-decomposition of 
%a given graph $G$ is NP-hard, as the case $k=2$ must be
%always checked. 

\begin{lemma}
For the literal and extended literal graph in 
Figure \ref{fig:literal} every cograph 2-decomposition
is a uniquely determined cograph 2-partition.

In particular, in every cograph 2-partition $\{E_1,E_2\}$
of the extended literal graph,
the edges of the triangle $(0,1,2)$ must be entirely
contained in one $E_i$ and the pending edge 
$[6,9]$ must be in the same edge set $E_i$ as the edges of the 
of the triangle. Furthermore, the edges $[9,10]$ and $[9,11]$ must be contained
in $E_j$, $i\neq j$. %, where $[6,9]\notin E_j$
\label{lem:col-lit}
\end{lemma}
\begin{proof}
It is easy to verify that the given cograph 2-partition $\{E_1,E_2\}$ in Fig.
\ref{fig:literal} fulfills the conditions and 
is correct, since $G=(V,E_1)$ and $G=(V,E_2)$ do not contain
induced $P_4$'s and are, thus, cographs.
We have to show that it is also unique.

Assume that there is another cograph 2-decomposition $\{F_1,F_2\}$. 
Note, for any cograph 2-decomposition $\{F_1,F_2\}$ it must hold that 
two incident edges in the triangle $(0,1,2)$ are contained in 
one of the sets $F_1$ or $F_2$. 
W.l.o.g.\ assume that $[0,1], [0,2]\in F_1$.

Assume first that  $[1,2] \not\in F_1$. 
In this case, because of the paths $P_4 = 6-2-0-1$ and $P_4 = 2-0-1-5$
it must hold that 
$[2,6], [1,5]\not\in F_1$ and thus, $[2,6], [1,5]\in F_2$. 
However, in this case and due to the paths $P_4 = 6-2-1-4$ and $2-0-1-4$ 
the edge $[1,4]$ can neither be contained in $F_1$ nor in $F_2$, 
a contradiction.  Hence, $[1,2] \in F_1$.

Note, the square $S_{1256}$ induced by vertices $1,2,5,6$ cannot have all
edges in $F_1$, as otherwise the subgraph $(V,F_1)$ would contain
the induced $P_4 = 6-5-1-0$. 
Assume that  
$[1,5]\in F_1$. As not all edges $S_{1256}$ are contained
in $F_1$, at least one of the edges $[5,6]$ and $[2,6]$ 
must be contained in $F_2$. If only one of the edges $[5,6]$, resp., $[2,6]$
is contained in $F_2$, we immediately obtain the 
induced $P_4=6-2-1-5$, resp., $6-5-1-2$ in $(V,F_1)$ and therefore, 
both edges $[5,6]$ and $[2,6]$ must be contained in $F_2$. 
But then the edge $[2,7]$ can neither be contained in $F_1$ (due to the induced $P_4=5-1-2-7$)
nor in $F_2$ (due to the induced $P_4=5-6-2-7$), a contradiction. 
Hence, $[1,5]\not\in F_1$ and thus, $[1,5]\in F_2$ for any $2$-decomposition. 
By analogous arguments and due to symmetry, all edges $[0,3]$, $[0,8]$, $[1,4]$, $[2,6]$, $[2,7]$
are contained in $F_2$, but not in $F_1$.

Moreover, due to the induced $P_4= 7-2-6-5$ and since $[2,6], [2,7] \in F_2$, 
the edge $[5,6]$ must be in $F_1$ and not in $F_2$. By analogous arguments and due to symmetry, 
it holds that $[3,4],[7,8]\in F_1$ and  $[3,4],[7,8]\not\in F_2$.
Finally, none of the edges of the triangle $(0,1,2)$ can be contained 
in $F_2$, as otherwise, we obtain an induced $P_4$ in $(V,F_2)$. 
Taken together, any $2$-decomposition of the literal graph must be a 
partition and is unique. 

Consider now the extended literal graph in Figure \ref{fig:literal}. As this
graph contains the literal graph as induced subgraph, the unique $2$-partition
of the underlying literal graph
is determined as by the preceding construction. 
Due to the path $P_4 = 7-2-6-9$ with $[2,6], [2,7] \in F_2$ we can 
conclude that $[6,9]\not\in F_2$ and thus $[6,9]\in F_1$. 
Since there are induced paths $P_4 = 5-6-9-y$, $y=10,11$ 
with $[5,6],[6,9]\in F_1$ we obtain that $[9,10], [9,11]\not\in F_1$
and thus, $[9,10], [9,11]\in F_2$ for any
$2$-decomposition (which is in fact a $2$-partition) 
of the extended literal graph, as claimed.
\end{proof}

\begin{figure}[t]
  \centering
  \includegraphics[bb= 0 475 398 660, scale=0.5]{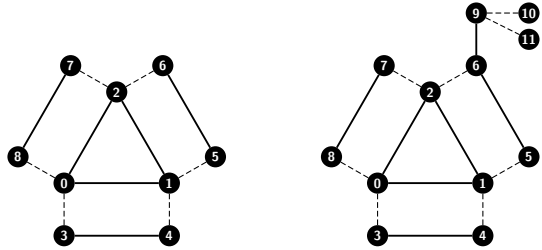}
  \caption{Left the \emph{literal graph} and right the \emph{extended} literal graph with 
			unique corresponding cograph 2-partition (indicated by dashed and bold-lined edges)
			is shown.}
  \label{fig:literal}
\end{figure}

\begin{figure}[t]
  \centering
  \includegraphics[bb=  4 322 554 611, scale=0.5]{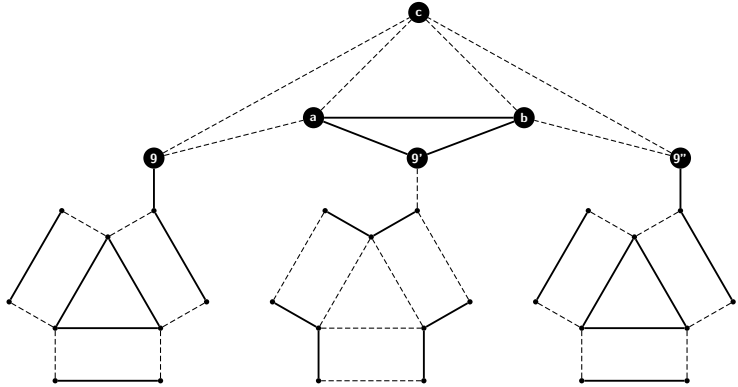}
  \caption{ Shown is a \emph{clause gadget} which consists of a triangle $(a,b,c)$ and
			 three extended literal graphs (as shown in Fig.\ \ref{fig:literal}) with edges attached to $(a,b,c)$. 
			A corresponding cograph 2-partition is indicated by dashed and bold-line edges. 
%			 The 2-coloring is unique (up to interchanging
%			 colors of the monochrom. triangles in literal gadgets - exactly one such mono-chromatic 
%			 triangle must have different colors as the two other mono-chromatic triangles.)
}
  \label{fig:clause}
\end{figure}

\begin{lemma}
	Given the clause gadget in Fig.\ \ref{fig:clause}. 
%	Any cograph 2-decomposition  $\{E_1,E_2\}$ of the clause gadget in Fig.\ \ref{fig:clause}.	
%	is a cograph 2-partition. 

	For any cograph 2-decomposition,  
  all edges of exactly two of the triangles in the underlying
	three extended literal graphs must be contained in one $E_i$
	and not in $E_j$,
	while the edges of the triangle of one extended literal graph must be in  $E_j$
	and not in $E_i$, $i\neq j$. 

	Furthermore, 
	for each cograph 2-decomposition exactly two of the edges $e,e'$
	of the triangle $(a,b,c)$ must be in one $E_i$ while the
	other edge $f$ is in $E_j$ but not in $E_i$, $j\neq i$. 
	The cograph 2-decomposition can be chosen so that 
	in addition 
	$e,e'\not \in E_j$, %and $f\not\in E_i$, $j\neq i$, 
	resulting in a cograph 2-partition of the 
	clause gadget.
%\TODO{nicht ganz siehe unten, brauchen wir das?}
%	If one fixes the coloring of the literal
%	gadgets, then the remaining 2-col of the graph in Fig.\ \ref{fig:clause}
%	is unique. 
	\label{lem:col-clause}
\end{lemma}
\begin{proof}
	It is easy to verify that the given cograph 2-partition in Fig. \ref{fig:clause} 
	fulfills the conditions and is correct, as $G=(V,E_1)$ and $G=(V,E_2)$ are cographs.  

%	For the sake of simplicity, we speak of an edge coloring $\de$ induced by
%	$E_1$ and $E_2$, i.e., $d(e)=d(f)$ for all $e,f\in E_i$, $i\in\{1,2\}$ and
%	 $d(e)\neq d(f)$ for all $e\in E_i,\ f\in E_j$, $i\neq j$. 

	As the clause gadget contains the literal graph as induced subgraph, the unique $2$-partition
	of the underlying literal graph is determined as by the construction given in Lemma \ref{lem:col-lit}. 
    Thus, each edge of the triangle in each underlying literal graph is contained in either one 
	of the sets $E_1$ or $E_2$. 
	Assume that edges of the triangles in the three literal gadgets are \emph{all} contained
	in the same set, say $E_1$. 
	Then, Lemma \ref{lem:col-lit} implies that 
	$[9,a], [9,c], [9',a],[9',b],[9'',b], [9'',c] \in E_1$
	and none of them is contained in $E_2$. 
	Since there are induced $P_4$'s: $9-a-b-9''$, $9'-a-c-9''$ and $9-c-b-9'$, 
	the edges 
	$[a,b], [a,c], [b,c]$ cannot be contained in $E_1$, and thus must be
	in $E_2$. However, this is 
	not possible, since then we would have the induced paths $P_4=9-a-9'-b$ in 
	the subgraph $(V,E_1)$
	a contradiction. Thus, the edges of the triangle of exactly one literal gadget 
    must be contained in a different set $E_i$ than
	the edges of the other triangles in the other two literal gadgets. 
	W.l.o.g. assume that the $2$-decomposition of the underlying 
	literal gadgets is given as in Fig.\ \ref{fig:clause}. 
	and identify bold-lined edges with $E_1$ and dashed edges with $E_2$. 
	
    It remains to show that this 2-decomposition of the underlying three 
	literal gadgets
	 determines %\TODO{in a unique way} 
	which of the edges
	of triangle $(a,b,c)$ are contained in which of the sets $E_1$ and $E_2$.	
%	(without necessarily inducing with edges  $(a,b,c)$ are not contained in
%	$E_1$ or $E_2$). 
	Due to the induced path $9-a-b-9''$ and since $[9,a],[9'',b]\in E_2$, the edge $[a,b]$ cannot be contained
	in $E_2$ and thus, is contained in $E_1$. 
%	\TODO{$[a,c]$,$[b,c]$ in $E_2$, klar, is aber nicht ausgeschlossen, dass die nicht auch in $E_1$ drin sind.
%	check, wo wir das brauchen im NP-beweis, ansonsten gibts hier halt degree-of-freedom, also 
%	sowohl ne decomposition, die keine partition, aber auch partition.}
	Moreover, if $[b,c]\not\in E_2$, then  
    then there is an induced path $P_4 = b-9''-c-9$ in the subgraph $(V,E_2)$, a
	contradiction. Hence,  $[b,c] \in E_2$ and by analogous arguments, $[a,c]\in E_2$. 
	If  $[b,c] \not\in E_1$ and  $[a,c] \not\in E_1$, then we
	obtain a cograph 2-partition. However, it can easily be verified that 
	there is still a degree of 
	freedom and $[a,c], [b,c]\in E_1$ is allowed for a 
	valid cograph 2-decomposition.
\end{proof}

We are now in the position to prove NP-completeness of \textsc{Cograph 2-Partition}
by reduction from the following problem. 

\begin{problem1}\textsc{Monotone NAE 3-SAT}\\[0.1cm]
\begin{tabular}{ll}
	  \emph{Input:}&Given a set $U$ of Boolean variables and a set of clauses \\
					&$\psi = \{C_1, \dots, C_m\}$ over $U$
					such that for all $i=1, \dots, m$ \\
					&it holds that $|C_i|=3$ and $C_i$ contains no negated variables. \\
	\emph{Question:}&Is there a truth assignment to $\psi$ such that in each $C_i$\\
					&not all three literals are set to true?
\end{tabular}
\end{problem1}

\begin{thm}[\cite{Sch78,moret-97}]
\textsc{Monotone NAE 3-SAT} is NP-complete. 
\end{thm}
%paper: On the complexity of the balanced vertex ordering problem
%http://pages.cs.wisc.edu/~tdw/files/slides/2013-05-17_Google-Madison.pdf
%http://tuvalu.santafe.edu/~moore/theory/hw4solns.pdf

% mon NAE 3 sat from nae-3 sat by replacing all literals
% of the form 1-v_i by y_i and adding additional clause 
% y_i,v_i,v_i (http://revealedpreferences.org/articles/1310635214main.pdf)

%http://www.cs.mcgill.ca/~sue/506/announcements/Schaefer.pdf

\begin{thm}
	\textsc{Cograph 2-Decomposition}, and thus, \textsc{Cotree 2-Representation} is NP-complete.
	\label{thm:npc}
\end{thm}
\begin{proof}
	Given a graph $G=(V,E)$ and cograph 2-decomposition $\{E_1,E_2\}$, 
	one can verify in linear time whether $(V,E_i)$ is a cograph \cite{corneil1985linear}.
	Hence, \textsc{Cograph 2-Partition} $\in$ NP. 

   We will show by reduction from \textsc{Monotone NAE 3-SAT} that
   \textsc{Cograph 2-Decomposition} is NP-hard. %, even for $k=2$. 
   Let $\psi = (C_1, \dots, C_m)$ be an arbitrary instance of
	\textsc{Monotone NAE 3-SAT}.
	Each clause $C_i$ is identified with a triangle $(a_i,b_i,c_i)$.
	Each variable $x_j$ is identified with a literal graph 
	as shown in Fig.\ \ref{fig:literal} (left) and different
	variables are identified with different literal graphs. 
	Let $C_i = (x_{i_1}, x_{i_2}, x_{i_3})$ and $G_{i_1}$,
	$G_{i_2}$ and $G_{i_3}$ the respective literal graphs. 
	Then, we extend each
	literal graph $G_{i_j}$ by adding an edge $[6,9_{i,j}]$. 
	Moreover, we add to $G_{i_1}$ the edges $[9_{i,1},a_i], [9_{i,1},c_i]$, 
	to $G_{i_2}$ the edges $[9_{i,2},a_i], [9_{i,2},b_i]$, 
	to $G_{i_3}$ the edges $[9_{i,3},c_i], [9_{i,3},b_i]$. The latter construction
	connects each literal graph with the triangle $(a_i,b_i,c_i)$ of
	the respective clause $C_i$ in a unique way, see Fig. \ref{fig:clause}.
	We denote the clause gadgets
	by $\Psi_i$ for each clause $C_i$. We repeat this construction
	for all clauses $C_i$ of $\psi$ resulting in the graph  $\Psi$. 
	An illustrative example is given in   Fig.\ \ref{fig:formula}.	
	Clearly, this reduction can be done in polynomial time
	in the number $m$ of clauses.  

	We will show in the following that $\Psi$ has a 
	cograph 2-decomposition (resp., a cograph 2-partition) 
	if and only if $\psi$ has a truth assignment $f$. 

	Let $\psi = (C_1, \dots, C_m)$ have a truth assignment. 
	Then in each clause $C_i$ at least one of the literals $x_{i_1}, x_{i_2}, x_{i_3}$ 
	is set to true and one to false. 
	We assign all edges $e$ of the triangle in the 
	corresponding literal graph $G_{i_j}$ to $E_1$, 
	if 	$f(x_{i_j})=true$ and to $E_2$, otherwise.
	Hence, each edge of exactly two of the triangles (one in $G_{i_j}$ and one in $G_{i_{j'}}$ 
	contained in one $E_r$	and not in $E_s$, 
	while the edges of the other triangle in $G_{i_{j''}}$, $j''\neq j,j'$
	are contained in  $E_s$	and not in $E_r$, $r\neq s$, 
	as needed
	for a possible valid cograph 2-decomposition (Lemma \ref{lem:col-clause}). 
	We now apply the construction of a valid 2-decomposition (or 2-partition) 
	for each $\Psi_i$ as given in Lemma \ref{lem:col-clause}, 
	starting with the just created assignment of edges contained in 	
	the triangles in $G_{i_j}$, $G_{i_{j'}}$ and $G_{i_{j''}}$ to $E_1$ or $E_2$. In this way, we obtain a valid 
	2-decomposition (or 2-partition) for each subgraph $\Psi_i$ of $\Psi$.
%	Lemma \ref{lem:col-clause} to obtain a valid
%	cograph 2-decomposition of each such clause gadget $\Psi_i$ in $\Psi$
%	by setting $E_1=\{e\in E(\Psi)\mid \de(e)=1\}$ and
%	$E_2=\{e\in E(\Psi)\mid \de(e)=0\} =E(\Psi)\setminus E_1$.	
	Thus, if there would be an induced $P_4$ in $\Psi$ with all edges 
	belonging to the same set $E_r$, 
	then this $P_4$ can only have edges belonging to different 
	clause gadgets $\Psi_k, \Psi_l$. 
	By construction, such a $P_4$ can only exist along 
	different clause gadgets $\Psi_k$ and $\Psi_l$ 
	only if $C_k$ and $C_l$  have a literal $x_i=x_{k_m}=x_{l_n}$ in 
	common. In this case, Lemma \ref{lem:col-clause} implies that 
	the edges $[6,9_{k,m}]$ and $[6,9_{l,n}]$ in $\Psi_i$ must belong
	to the same set $E_r$. % and obtain therefore the same color. 
	Again by Lemma \ref{lem:col-clause}, the edges $[9_{k,m},y]$
	and  $[9_{k,m},y']$,\ $y,y'\in \{a_k,b_k,c_k\}$ as well as
	the edges $[9_{l,n},y]$
	and  $[9_{l,n},y']$,\ $y,y'\in \{a_l,b_l,c_l\}$ must be
	in a different set $E_s$ than $[6,9_{k,m}]$ and $[6,9_{l,n}]$. 
	Moreover, respective edges $[5,6]$ in $\Psi_k$, as well as in 
	 $\Psi_l$ (Fig.\ \ref{fig:literal}) must then 
	be in $E_r$, i.e., in the same set as $[6,9_{k,m}]$ and $[6,9_{l,n}]$. 
%	the same color $(6,9_{k_m})$. The same is true for the edges
%	$(5,6)$,  $(6,9_{l_n})$ in $\Psi_l$. 
	However, in none of the cases it is possible to find
	an induced $P_4$ with all edges 
	in the same set $E_r$ or $E_s$ along different clause gadgets. 
	Hence, we 
%	\TODO{We finally assign edges [x,y], $x,y\in \{a_l,b_l,c_l\}$ l=1...k
%	to $E_1$ or $E_2$, depending on [9x],[9y] edges to }
	obtain a valid cograph 2-decomposition, resp., cograph 2-partition of $\Psi$. 
%To obtain a valid 2-partition of 
%	$\Psi$ one can again use th
%	depending on the coloring 

\begin{figure}[t]
  \centering
  \includegraphics[bb=  9 490 607 679, scale=0.55]{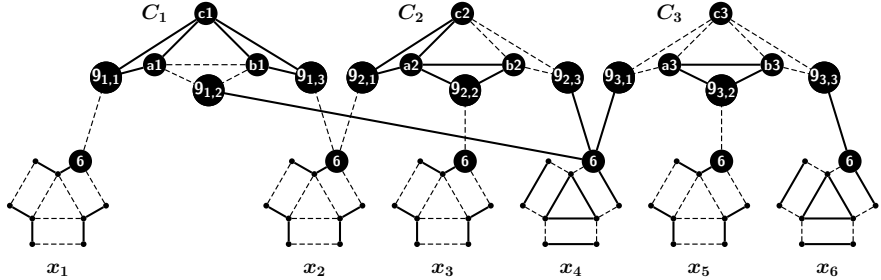}
  \caption{Shown is the graph $\Psi$ as constructed in the proof of Theorem \ref{thm:npc}.
			In particular, $\Psi$ reflects the NAE 3-SAT formula 
  			$\psi = \{C_1, C_2, C_3\}$ with clauses $C_1=(x_1,x_4,x_2), C_2=(x_2,x_3,x_4)$ and 
  				$C_3=(x_4,x_5,x_6)$. Different literals obtain the same truth assignment true or false, 
				whenever the edges of the triangle in their corresponding literal gadget are contained in 
				the same set $E_i$ of the cograph 2-partition, highlighted by dashed and bold-lined edges.  }
  \label{fig:formula}
\end{figure}

	Now assume that $\Psi$ has a valid cograph 2-decomposition (or a
	2-partition). Any variable $x_{i_j}$ contained in some clause $C_i =
	(x_{i_1}, x_{i_2}, x_{i_3})$ is identified with a literal graph
	$G_{i_j}$. Each clause $C_i$ is, by construction, identified with
	exactly three literal graphs $G_{i_1},G_{i_2},G_{i_3}$, resulting in
	the clause gadget $\Psi_i$. Each literal graph $G_{i_j}$ contains
	exactly one triangle $t_j$. Since $\Psi_i$ is an induced subgraph of
	$\Psi$, we can apply Lemma \ref{lem:col-clause} and conclude that for
	any cograph 2-decomposition (resp., 2-partition) all edges of exactly
	two of three triangles $t_1,t_2,t_3$ are contained in one set $E_r$,
	but not in $E_s$, and all edges of the other triangle are contained in
	$E_s$, but not in $E_r$, $s\neq r$. Based on these triangles we define
	a truth assignment $f$ to the corresponding literals: w.l.o.g.\ we
	set $f(x_i)=$\emph{true} if the edge $e\in t_i$ is contained in $E_1$
	and $f(x_i)=$\emph{false} otherwise. By the latter arguments and Lemma
	\ref{lem:col-clause}, we can conclude that, given a valid cograph
	2-partitioning, the so defined truth assignment $f$ is  a valid truth
	assignment of the Boolean formula $\psi$, since no three different literals in one
	clause obtain the same assignment and at least one of the variables is
	set to $\emph{true}$. Thus, \textsc{Cograph 2-Decomposition} is
	NP-complete
	
	Finally, because \textsc{Cograph 2-Decomposition} and \textsc{Cotree 2-Representation} 
	are equivalent problems, 
	the NP-completeness of \textsc{Cotree 2-Representation} follows. 
\end{proof}

As the proof of Theorem \ref{thm:npc} allows us to use cograph 2-partitions in all proof steps,
instead of cograph 2-decompositions, we can immediately infer the
NP-completeness of the following problem for k=2, as well. 

%Note, the latter proof 
% in NP-completeness proof    thm:ref,
%only constructs non-cographs that have a 2-partition. 
%Hence, 
%allows to infer that the next problem is NP-complete, as well
%To only remaining degrees of freedom are to color the triangle ai,bi,ci ... 
%it is easy that nae3sat true iff constructed graph has a cograph-2-partition

\begin{problem1}\textsc{Cograph $k$-Partition}\\[0.1cm]
\begin{tabular}{ll}
	 \emph{Input:}&Given a graph $G=(V,E)$ and an integer $k$. \\
	 \emph{Question:}&Is there a Cograph $k$-Partition of $G$?
\end{tabular}
\end{problem1}

\begin{thm}
	\textsc{Cograph 2-Partition} is NP-complete. 
\end{thm}

As a direct consequence of the latter results, we obtain the 
following theorem. 

%As the putative solutions of the 
% problems \textsc{Cograph K-Partition}, \textsc{Cograph $k$-Decomposition}, \textsc{Cotree $k$-Representation}
%can be checked in to be valid or not in polynomial time (\TODO{citieren .. }) these three problems
%are contained in NP .. but why NP-complete ?? 

\begin{thm}
Let $G$ be a given graph that is not a cograph. 
The following three optimization problems to find the least integer $k>1$ so that 
there is  a Cograph $k$-Partition, or a Cograph $k$-Decomposition, or
a Cotree $k$-Representation
%\begin{itemize}
%	\item a Cograph $k$-Partition, or \vspace{-0.4em}
%	\item a Cograph $k$-Decomposition, or\vspace{-0.4em}
%	\item a Cotree $k$-Representation
%\end{itemize}
for the graph $G$, are NP-hard. 
%\textsc{Cograph K-Partition}, \textsc{Cograph $k$-Decomposition}, \textsc{Cotree $k$-Representation}
%are NP-complete... mmh ???  da brauchen wir noch ne reduction von 2-part oder so .. 
%The three optimization problems, to 
%find the least integer $K\geq 2$ so that 
%\textsc{Cograph K-Partition}, \textsc{Cograph $k$-Decomposition}, \textsc{Cotree $k$-Representation}
%are fullfilled, are NP-hard.
\end{thm}

\section*{Acknowledgment}
This work was funded by the German Research Foundation
(DFG) (Proj. No. MI439/14-1).

%
% ---- Bibliography ----
%
%\section*{References	}

\bibliographystyle{plain}	
\bibliography{paper}

\end{document}